\documentclass{article}
\usepackage{fullpage}

\usepackage{authblk}
\usepackage{microtype}
\usepackage{amsthm}
\usepackage{amssymb,amsmath}
\usepackage{braket}
\usepackage{microtype}
\usepackage{color}
\usepackage{graphicx}
\usepackage{comment}
\usepackage{caption}
\usepackage{enumitem}
\usepackage{array}
\usepackage{refcount}

\usepackage{hyperref}
\hypersetup{
    pdfborder={0 0 0},
}

\renewcommand{\paragraph}[1]{\medskip\noindent\textbf{#1}}

\newtheorem{theorem}{Theorem}

\newtheorem{lemma}{Lemma}

\newcommand{\rl}[1]{{\footnotesize\color{red}[Ling: #1]}}

\newcommand{\ignore}[1]{}

\newcommand{\sig}[1]{\braket{#1}}
\newcommand{\Sign}[2]{\mathsf{Sign}_{#1}\left({#2}\right)}

\newcommand{\Accepted}{\mathsf{accepted}}

\newcommand{\Status}{\mathsf{status}}
\newcommand{\Propose}{\mathsf{propose}}

\newcommand{\Commit}{\mathsf{commit}}
\newcommand{\Notify}{\mathsf{notify}}
\newcommand{\Blame}{\mathsf{view\text{-}change}}
\newcommand{\NewView}{\mathsf{new\text{-}view}}
\newcommand{\Cert}{\mathcal{C}}
\newcommand{\NotSum}{\mathcal{N}}
\newcommand{\BlameSum}{\mathcal{V}}
\newcommand{\CPSum}{\mathcal{S}}

\newcommand{\NotifyLight}{\mathsf{notify\text{-}light}}
\newcommand{\StatusMax}{\mathsf{status\text{-}max}}
\newcommand{\Enter}{\mathsf{enter}}
\newcommand{\drift}{t}
\newcommand{\Sync}{\mathsf{sync}}
\newcommand{\SetTime}{\mathsf{new\text{-}day}}
\newcommand{\vLi}{v_{L \rightarrow i}}

\begin{document}

\title{Efficient Synchronous Byzantine Consensus}

\author[1]{Ittai Abraham}
\author[2]{Srinivas Devadas}
\author[3]{Danny Dolev}
\author[4]{Kartik Nayak}
\author[2]{Ling Ren}

\affil[1]{VMware Research -- \texttt{iabraham@vmware.com}}
\affil[2]{Massachusetts Institute of Technology -- \texttt{\{devadas, renling\}@mit.edu}}
\affil[3]{The Hebrew University of Jerusalem -- \texttt{dolev@cs.huji.ac.il}}
\affil[4]{University of Maryland, College Park -- \texttt{kartik@cs.umd.edu}}


\maketitle

\begin{abstract}
We present new protocols for Byzantine state machine replication and Byzantine agreement in the synchronous and authenticated setting.
The celebrated PBFT state machine replication protocol tolerates $f$ Byzantine faults in an asynchronous setting using $3f+1$ replicas, 
and has since been studied or deployed by numerous works.
In this work, we improve the Byzantine fault tolerance threshold to $n=2f+1$ by utilizing a relaxed synchrony assumption.
We present a synchronous state machine replication protocol that commits a decision every 3 rounds in the common case.
The key challenge is to ensure quorum intersection at one honest replica.
Our solution is to rely on the synchrony assumption to form a post-commit quorum of size $2f+1$, 
which intersects at $f+1$ replicas with any pre-commit quorums of size $f+1$.
Our protocol also solves synchronous authenticated Byzantine agreement in expected 8 rounds.
The best previous solution (Katz and Koo, 2006) requires expected 24 rounds.
Our protocols may be applied to build Byzantine fault tolerant systems or improve cryptographic protocols such as cryptocurrencies when synchrony can be assumed. 
\end{abstract}

\section{Introduction}	\label{sec:intro}

Byzantine consensus~\cite{LSP82,PBFT} is a fundamental problem in distributed computing and cryptography.
It has been used to build fault tolerant systems such as distributed storage systems~\cite{BASE01,PBFT02,Farsite,Oceanstore,Abd05,HQRep06,Kotla04}, certificate authorities~\cite{Omega96,COCA02}, fair peer-to-peer sharing~\cite{Wallach03}, and more recently cryptocurrencies~\cite{Byzcoin,Algorand,HybridConsensus,Solidus}.
It has also been frequently used as building blocks in cryptographic protocols such as secure multi-party computation~\cite{GMW87,BGW88}. 

Broadly speaking, Byzantine consensus considers the problem of reaching agreement among a group of $n$ parties, among which up to $f$ can have Byzantine faults and deviate from the protocol arbitrarily.  
There exist a few variant formulations for the Byzantine consensus problem.\footnote{We use the word ``consensus'' as a collective term for these variants; other papers have different conventions.} 
Two theoretical formulations are Byzantine broadcast and Byzantine agreement~\cite{PSL80,LSP82}.
In Byzantine broadcast, there is a designated \emph{sender} who tries to broadcast a value;
In Byzantine agreement, every party holds an initial input value.
To rule out trivial solutions, both problems have additional validity requirements.
Byzantine broadcast and agreement have been studied under various combinations of timing (synchrony, asynchrony or partial synchrony) and cryptographic assumptions (whether or not to assume digital signatures).
It is now well understood that these assumptions drastically affect the bounds on fault tolerance. 
In particular, Byzantine agreement requires $f<n/3$ under partial synchrony or asynchrony even with digital signatures,
but can be solved with $f<n/2$ under synchrony with digital signatures.

Most Byzantine broadcast and agreement protocols have been designed to demonstrate theoretical feasibility and the problem definitions are also not always convenient to work with in practice. 
A more practice-oriented problem formulation is Byzantine fault tolerant (BFT) state machine replication~\cite{Schneider90,PBFT}.
In this formulation, the goal is to design a replicated service that provides the same interface as a single server, despite some replicas experiencing Byzantine faults.
In particular, honest replicas agree on a sequence of values and their \emph{order}, while the validity of the values is left outside the protocol. 
The PBFT protocol by Castro and Liskov~\cite{PBFT} is an asynchronous state machine replication protocol that tolerates $f<n/3$ Byzantine faults.  
As the first BFT protocol designed for practical efficiency, PBFT has since inspired numerous follow-up works including many practical systems~\cite{BASE01,Farsite,Oceanstore,COCA02,Wallach03,Yin03,Abd05,HQRep06,Martin06,Zyzzyva}. 

Perhaps somewhat surprisingly, we do not yet have a practical solution for Byzantine consensus in the seemingly easier synchronous and authenticated (i.e., with digital signatures) setting.
To the best of our knowledge, the most efficient Byzantine agreement protocol with the optimal $f<n/2$ fault tolerance in this setting is due to Katz and Koo~\cite{KK06}, 
which requires in expectation 24 rounds of communication (not counting the random leader election subroutine). 
To agree on many messages sequentially, it requires additional generic transformations~\cite{Lindell02,KK06} that further increase the expected round complexity to a staggering 72 rounds per instance!
The only state machine replication protocol we know of in this setting is XFT~\cite{XFT}.
Relying on an active group of $f+1$ honest replicas to make progress,
XFT is very efficient for small $n$ and $f$ (e.g., $f=1$).
But its performance degrades as $n$ and $f$ increase, especially when the number of faults $f$ approaches the $\lfloor \frac {n-1}2 \rfloor$ limit.
In that case, among the ${n \choose f+1}$ groups in total, only one is all-honest.
The simplest variant of XFT (presented in~\cite{XFT}) requires an exponential number of view changes to find that group.
The best XFT variant we can think of still requires $\Theta(n^2)$ view changes to find that group 
(we describe such a variant in Appendix~\ref{sec:xft2}).

This paper presents efficient Byzantine consensus protocols for the synchronous and authenticated setting tolerating $f<n/2$ faults.
Our main focus is BFT state machine replication, for which our protocol requires amortized 3 rounds per decision.
In scenarios where synchrony can be assumed, our protocol can be applied to build BFT systems and services tolerating $f<n/2$ Byzantine faults, improving upon the $f<n/3$ fault threshold of PBFT-style protocols. 
Meanwhile, our protocol can also solve multi-valued Byzantine broadcast and agreement for $f<n/2$ in expected 8 rounds assuming a random leader oracle.
We also remark that we do not need to assume that replicas act in locked step synchronized rounds;
rather, we present a simple clock synchronization protocol to bootstrap locked step synchrony from bounded message delay and bounded clock drift.

\subsection{Overview of Our Protocols}

Interestingly, our core protocol draws inspiration from the Paxos protocol~\cite{Paxos}, which is neither synchronous nor Byzantine fault tolerant.
Since our main focus is state machine replication, we will describe the core protocol with ``replicas'' instead of ``parties''.
The core of our protocol resembles the synod algorithm in Paxos, but is adapted to the synchronous and Byzantine setting.
In a nutshell, it runs in iterations with a unique leader in each iteration (how to elect leaders is left to higher level protocols).
Each new leader picks up the states left by previous leaders and drives agreement in its iteration.  
A Byzantine leader can prevent progress but cannot violate safety.
As soon as an honest leader emerges, then all honest replicas reach agreement and terminate at the end of that iteration.

While synchrony is supposed to make the problem easier, it turns out to be non-trivial to adapt the synod algorithm to the synchronous and Byzantine setting while achieving the optimal $f<n/2$ fault tolerance.
The major challenge is to ensure \emph{quorum} intersection~\cite{Paxos} at one \emph{honest} replica.
The core idea of Paxos is to form a quorum of size $f+1$ before a commit. 
With $n=2f+1$, two quorums always intersect at one replica, which is honest in Paxos. 
This honest replica in the intersection will force a future leader to respect the committed value.
In order to tolerate $f$ Byzantine faults, PBFT uses quorums of size $2f+1$ out of $n=3f+1$, so that two quorums intersect at $f+1$ replicas, among which one is guaranteed to be honest. 
At first glance, our goal of one honest intersection seems implausible with the $n=2f+1$ constraint.
Following PBFT, we need two quorums to intersect at $f+1$ replicas which seems to require quorums of size $1.5f+1$. 
On the other hand, a quorum size larger than $f+1$ (the number of honest replicas) seems to require participation from Byzantine replicas and thus loses liveness. 
Our solution is to utilize the synchrony assumption to form a \emph{post-commit quorum} of size $2f+1$.
A post-commit quorum does not affect liveness and intersects with any \emph{pre-commit quorum} (of size $f+1$) at $f+1$ replicas.  
This satisfies the requirement of one honest replica in intersection. 
With some additional checks and optimizations, we obtain our core protocol: a 4-round synchronous Byzantine synod protocol (three Paxos-like rounds plus a notification round). 
It preserves safety under Byzantine leaders and ensures termination once an honest leader emerges.

We then apply the core synod protocol to state machine replication and Byzantine broadcast/agreement in the synchronous and authenticated setting with $f<n/2$.
For state machine replication, a simple strategy is to rotate the leader role among the replicas after each iteration.
Because each honest leader is able to drive at least one decision, the protocol spends amortized 2 iterations (8 rounds) per decision with $f<n/2$ faults.
We then improve the protocol to allow a stable leader and only replace the leader if it is not making progress.
The improved protocol commits a decision in 3 rounds in the common case.
While our view change protocol resembles that of PBFT at a high level,  
the increased fault threshold $f<n/2$ again creates new challenges.
In particular, two views in PBFT cannot make progress concurrently: 
$f+1$ honest replicas need to enter the new view to make progress there, leaving not enough replicas for a quorum in the old view. 
In contrast, with a quorum size of $f+1$ and $n=2f+1$ in our protocol, if a single honest replica is left behind in the old view, the $f$ Byzantine replicas can exploit it to form a quorum. 
Thus, our view change protocol needs to ensure that two honest replicas are never in different views. 
Informally, our protocol achieves the following result.
More precise results are given in Appendix~\ref{sec:smr:stable_leader}.

\begin{theorem}
(Informal) There exists a synchronous leader-based state machine replication protocol for $n=2f+1$. 
Each decision takes 3 rounds in the common case.
View changes (replacing a leader) take 4 rounds and happen at most $f$ times.  
\label{thm:smr}
\end{theorem}

To solve Byzantine broadcast, we let the designated sender be the leader for the first iteration.
After the first iteration, we rotate the leader role among all $n$ parties.
It is straightforward to see that this solution achieves both agreement and validity.
If the designated sender is honest, every honest party agrees on its value and terminates. 
Otherwise, the first honest leader that appears down the line will ensure agreement and termination for all honest parties.
Assuming we have a random leader oracle, there is a $(f+1)/(2f+1) > 1/2$ probability that each leader after the first iteration is honest, so the protocol terminates in expected 2 iterations after the first iteration. 
To solve Byzantine agreement, we can use the classical transformation from Lamport et al.~\cite{LSP82}.
These give rise to the results in Theorem~\ref{thm:ba}.

\begin{theorem}
Assuming a random leader election oracle, 
there exist synchronous Byzantine broadcast and Byzantine agreement protocols for $f<n/2$ that terminate in expected 8 rounds. 
\label{thm:ba}
\end{theorem}

We remark that the $f<n/2$ Byzantine fault tolerance in our protocols is optimal for synchronous authenticated Byzantine agreement and state machine replication, but not for Byzantine broadcast. 
Our quorum-based approach cannot solve Byzantine broadcast in the dishonest majority case ($f \geq n/2$).


\section{Related Work}	\label{sec:related} 

\paragraph{Byzantine agreement and broadcast.}
The Byzantine agreement and Byzantine broadcast problems were first introduced by Lamport, Shostak and Pease~\cite{LSP82, PSL80}.
They presented protocols and fault tolerance bounds for two settings (both synchronous).
Without cryptographic assumptions (the unauthenticated setting), Byzantine broadcast and agreement can be solved if $f<n/3$.
Assuming digital signatures (the authenticated setting), Byzantine broadcast can be solved if $f<n$ and Byzantine agreement can be solved if $f<n/2$.
The initial protocols had exponential message complexities~\cite{PSL80,LSP82}.
Fully polynomial protocols were later shown for both the authenticated ($f<n/2$)~\cite{DS83} and the unauthenticated ($f<n/3$)~\cite{Garay98} settings.
Both protocols require $f+1$ rounds of communication, 
which matches the lower bound on round complexity for deterministic protocols~\cite{FL82}. 
To circumvent the $f+1$ round lower bound,
a line of work explored the use of randomization~\cite{BenOr83,Rabin83} which eventually led to expected constant-round protocols for both the authenticated ($f<n/2$)~\cite{KK06} and the unauthenticated ($f<n/3$)~\cite{FM97} settings. 
In the asynchronous setting, the FLP impossibility~\cite{FLP85} rules out any deterministic solution. 
Some works use randomization~\cite{BenOr83,Bracha87} or partial synchrony~\cite{DLS88} to circumvent the impossibility. 
 
\paragraph{State machine replication.}
A more practical line of work studies state machine replication~\cite{Lamport78,Schneider90}.
The goal is to design a distributed system consisting of replicas to process requests from external clients while behaving like a single-server system.
Paxos~\cite{Paxos} and Viewstamped replication~\cite{Viewstamped} tolerate $f$ crash faults with $n \geq 2f+1$ replicas. 
The PBFT protocol~\cite{PBFT} tolerates $f$ Byzantine faults with $n \geq 3f+1$ replicas. 
In all three protocols, safety is preserved even under asynchrony while progress is made only during synchronous periods.
Numerous works have extended, improved or deployed PBFT~\cite{BASE01,PBFT02,Farsite,Oceanstore,COCA02,Wallach03,Yin03,Abd05,HQRep06,Martin06,Zyzzyva}.
They all consider the asynchronous setting and require $n \geq 3f+1$. 
Several systems achieve BFT state machine replication with $n \geq 2f+1$ by introducing trusted components to the protocol~\cite{Correia04,Chun07,CheapBFT12,Veronese13}. 
To the best of our knowledge, the only work on state machine replication that considers the exact same setting as ours (Byzantine faults, $n \geq 2f+1$, synchrony, digital signatures and no trusted component) is XFT~\cite{XFT}.
We remark that the main goal of XFT is to tolerate \emph{either} Byzantine faults under synchrony \emph{or} crash faults under asynchrony,
but we can still compare to its synchronous Byzantine version. 
As we mentioned, with the best techniques we are aware of, XFT's performance does not scale well with $n$ and $f$.




\section{A Synchronous Byzantine Synod Protocol}	\label{sec:protocol}

\subsection{Model and Overview}

Our core protocol is a synchronous Byzantine synod protocol with $n=2f+1$ replicas.
An adversary may corrupt up to $f$ replicas and may adaptively decide which replicas to corrupt as the protocol proceeds. 
The adversary is not \emph{mobile} and cannot ``uncorrupt'' replicas; the total number of replicas that the adversary has ever corrupted is at most $f$.
Corrupted replicas are coordinated by the adversary and may deviate from the protocol arbitrarily.
The goal of the core synod protocol is to guarantee that all honest
replicas eventually commit (liveness) and commit on the same value
(safety).
Note that we use the term \emph{honest} for a node that is not faulty
whereas a faulty node is referred to as \emph{Byzantine}.

We assume synchrony.
If an honest replica $i$ sends a message to another honest replica $j$ at the beginning of a round, the message is guaranteed to reach by the end of that round. 
Our protocol runs in iterations and each iteration consists of 4 rounds.
We describe the protocol assuming all replicas have perfectly synchronized clocks.
Hence, they enter each round simultaneously and have the same view on the current iteration number $k$ (the first iteration has $k=1$).
In practice, a known bound on the communication delay is sufficient.

We assume public key cryptography. 
Every replica knows the public (verification) key of every other replica, and they use digital signatures when communicating with each other.
Byzantine replicas cannot forge honest replicas' signatures, which means messages in our systems enjoy authenticity as well as non-repudiation.
We use $\sig{x}_i$ to denote a message $x$ that is signed by replica $i$,
i.e., $\sig{x}_i = (x, \sigma)$ where $\sigma = \Sign{i}{x}$ is a signature produced by replica $i$ using its private signing key.
For better efficiency, $\sigma$ can be a signature of a message's digest, i.e., output of a collision resistant hash function. 
A message can be signed by multiple replicas (or the same replica) in layers, i.e., 
$\sig{\sig{x}_i}_j = \sig{x, \sigma_i}_j = (x, \sigma_i, \sigma_j)$ where $\sigma_i = \Sign{i}{x}$ and $\sigma_j = \Sign{j}{x~||~\sigma_i}$ ($||$ denotes concatenation). 

The core protocol assumes a unique leader in each iteration that is known to every replica.
An iteration leader can be one of the replicas but can also be an external entity.
Similar to Paxos, we decouple leader election from the core protocol and leave it to higher level protocols (Section~\ref{sec:smr}~and~\ref{sec:BABB}) or, in some cases, the application level.
For example, a cryptocurrency (blockchain) may elect leaders based on proof of work.
 
Each iteration consists of 4 rounds.
The first three rounds are conceptually similar to Paxos:
(1) the leader learns the states of the system, 
(2) the leader proposes a safe value, and 
(3) every replica sends a commit request to every other replica.
If a replica receives $f+1$ commit requests for the same value, it commits on that value.
If a replica commits, it notifies all other replicas about the commit using a 4th round.
Upon receiving a notification, other replicas \emph{accept} the committed value and will vouch for that value to future leaders.
To tolerate Byzantine faults, we need to add equivocation checks and other proofs of honest behaviors at various steps.
We now describe the protocol in detail.

\subsection{Detailed Protocol}
\label{sec:main-protocol}

Each replica $i$ internally maintains some long-term states $\Accepted_i = (v_i, k_i, \Cert_i)$ across iterations to record its \emph{accepted} value.
Initially, each replica $i$ initializes $\Accepted_i := (\bot, 0, \bot)$. 
If replica $i$ later accepts a value $v$ in iteration $k$, it sets $\Accepted_i := (v, k, \Cert)$ such that $\Cert$ \emph{certifies} that $v$ is legally accepted in iteration $k$ (see Table~\ref{tab:valid}). 
In the protocol, honest replicas will only react to \emph{valid} messages. Invalid messages are simply discarded. 
To keep the presentation simple, we defer the validity definitions of all types of messages to Table~\ref{tab:valid}. 
We first describe the protocol assuming no replica has terminated, and later amend the protocol to deal with non-simultaneous termination.

\begin{itemize}
\addtolength{\itemindent}{32pt}
\item[\textbf{Round 0}] ($\Status$)
Each replica $i$ sends a $\sig{\sig{k, \Status, v_i, k_i}_i, \Cert_i}_i$ message to the leader $L_k$ of the current iteration $k$, 
informing $L_k$ of its current accepted value.
We henceforth write $L_k$ as $L$ for simplicity.

At the end of this round, the leader $L$ picks $f+1$ valid $\Status$ messages to form a \emph{safe value proof} $P$. 

\item[\textbf{Round 1}] ($\Propose$)
The leader $L$ picks a value $v$ that is \emph{safe to propose} under $P$:
$v$ should match the value that is accepted in the most recent iteration in $P$,
or any $v$ is safe to propose if no value has been accepted in $P$ (see Table~\ref{tab:valid} for more details).
$L$ then sends a signed proposal $\sig{\sig{k, \Propose, v}_L, P}_L$ to all replicas including itself. 

At the end of this round, if replica $i$ receives a valid proposal 
$\sig{\sig{k, \Propose, v}_L, P}_L$ from the leader,  it sets $\vLi := v$.
Otherwise (leader is faulty), it sets $\vLi := \bot$. 

\item[\textbf{Round 2}] ($\Commit$)
If $v := \vLi \neq \bot$, replica $i$ forwards the proposal $\sig{k, \Propose, v}_L$ (excluding $P$) and sends a $\sig{k, \Commit, v}_i$ request to all replicas including itself.

At the end of this round, if replica $i$ is forwarded a valid proposal $\sig{k, \Propose, v'}_L$ in which $v' \neq \vLi$, it does not commit in this iteration (leader has equivocated).
Else, if replica $i$ receives $f+1$ valid $\sig{k, \Commit, v}_j$ requests in all of which $v = \vLi$, 
it commits on $v$ and sets its long-term state $\Cert_i$ to be these $f+1$ $\Commit$ requests concatenated.
In other words, replica $i$ commits if and only if it receives $f+1$ valid and matching $\Commit$ requests and does not detect leader equivocation. 

\item[\textbf{Round 3}] ($\Notify$)
If replica $i$ has committed on $v$ at the end of Round 2,
it sends a notification $\sig{\sig{k, \Notify, v}_i, \Cert_i}_i$ to every other replica, and terminates.

At the end of this round, if replica $i$ receives a valid $\sig{\sig{k, \Notify, v}_j, \Cert}_j$ message, 
it accepts $v$ by setting its long-term states $\Accepted_i = (v_i, k_i, \Cert_i):= (v, k, \Cert)$.
If replica $i$ receives multiple valid $\Notify$ messages with different values, it is free to accept any one or none.
Lastly, replica $i$ increments the iteration counter $k$ and enters the next iteration.
\end{itemize}

\paragraph{Summaries and certificates.}
In a $\sig{\sig{k, \Status, v_i, k_i}_i, \Cert_i}_i$ message, we call the $\sig{k, \Status, v_i, k_i}_i$ component a $\Status$ summary, and the $\Cert_i$ component a $\Commit$ \emph{certificate}.
Similarly, in a $\sig{\sig{k, \Notify, v}_i, \Cert_i}_i$ message, we call the $\sig{k, \Notify, v}_i$ component a $\Notify$ summary, and $\Cert_i$ is again a $\Commit$ certificate.
This distinction will be important soon for handling non-simultaneous termination. 

\paragraph{A shorter safe value proof $P$.}
In the above basic protocol, $P$ consists of $f+1$ valid $\Status$ messages, each of which contains a $\Commit$ certificate.
This yields a length of $|P| = O(n^2)$.
We observe that $P$ can be optimized to have length $O(n)$.
$P$ can consist of $f+1$ valid $\Status$ summaries plus a $\Commit$ certificate for a summary that claims the highest iteration number.
Table~\ref{tab:valid} presents more details.

\paragraph{Non-simultaneous termination.}
We need to ensure that all honest replicas eventually commit and terminate.
However, in the protocol above, it is possible that some honest replicas terminate in an iteration while other honest replicas enter the next iteration.
Without special treatment, the honest replicas who enter the new iteration will never be able to terminate 
(unable to gather $f+1$ matching $\Commit$ requests) if Byzantine replicas simply stop participating.
To solve this problem, we use a standard idea: honest replicas continue participating ``virtually'' after termination. 

If replica $t$ has committed on $v_t$ and terminated in iteration $k_t$, 
it is supposed to send a valid notification $\sig{\sig{k_t, \Notify, v_t}_t, \Cert_t}_t$ to all other replicas.
Upon receiving this notification, replica $i$ becomes aware that replica $t$ has terminated, and does not expect any messages from replica $t$ in future iterations.
Of course, it is also possible that replica $t$ is Byzantine and sends notifications only to a subset of replicas.
We need to ensure that such a fake termination does not violate safety or liveness. 

Assuming that replica $i$ has received a valid notification $\sig{\sig{k_t, \Notify, v_t}_t, \Cert_t}_t$ from replica $t$,
we now amend the protocol such that in all future iterations $k > k_t$, replica $i$ ``pretends'' that it keeps receiving \emph{virtual} messages from replica $t$ in the follow ways:
\begin{itemize} 
\item[--] In \textbf{Round 0}, 
	if replica $i$ is the current iteration leader, 
	it treats the notification $\sig{\sig{k_t, \Notify, v_t}_t, \Cert_t}_t$ as a valid $\sig{\sig{k, \Status, v_t, k_t}_t, \Cert_t}_t$ message from replica $t$.
	In particular, the safe value proof $P$ is allowed to include the $\Notify$ summary $\sig{k_t, \Notify, v_t}_t$ in place of a $\Status$ summary.

\item[--] In \textbf{Round 1}, 
	if replica $t$ is the current iteration leader $L$, 
	then replica $i$ treats the $\Notify$ summary $\sig{k_t, \Notify, v_t}_t$ as a virtual proposal for $v_t$.
	If $v_t = v_i$ (the value replica $i$ accepts), then replica $i$ considers the virtual proposal valid sets $\vLi=v_t$.
	Later in Round 2, replica $i$ forwards the virtual proposal to all replicas for equivocation checking. 

\item[--] In \textbf{Round 2}, 
	replica $i$ treats the $\Notify$ summary $\sig{k_t, \Notify, v_t}_t$ as a valid $\Commit$ request for $v_t$ from replica $t$.
	In particular, a $\Commit$ certificate $\Cert_i$ is allowed to include this $\Notify$ summary in place of a $\Commit$ request.
\end{itemize}

\begin{table}[bt]
\centering
\caption{
\textbf{Validity requirements of messages.}
Every message must carry the current iteration number $k$ and be signed by the sender. 
Additional validity requirements are listed below.
A message's validity may depend on the validity of its components and other conditions defined in the table.
}
\label{tab:valid}
\renewcommand{\arraystretch}{1.3}
\begin{tabular}{ >{\centering}m{4cm} | m{10cm} }
\hline
Message & \multicolumn{1}{c}{Validity requirements} \\ 
\hline
\hline
$\sig{\sig{k, \Status, v_i, k_i}_i, \Cert_i}_i$
	&
	$k_i=0$ (initial state) or $\Cert_i$ certifies $(v_i, k_i)$ if $k_0 > 0$. \\ 
\hline

$\sig{\sig{k, \Propose, v}_L, P}_L$
	& $v$ is safe to propose under $P$.  
	$L$ is the leader of iteration $k$. \\
\hline

$\sig{k, \Commit, v}_i$ 
	& No extra requirement. \\
\hline

$\sig{\sig{k, \Notify, v}_i, \Cert}_i$ 
	& $\Cert$ certifies $(v, k)$. \\
\hline
\hline

$v$ is safe to propose under $P$
	& 
	$P$ consists of $f+1$ $\Status$ or $\Notify$ summaries, plus a $\Commit$ certificate for a summary that claims the highest non-zero iteration number.
	Each $\Status$ summary in $P$ must carry the current iteration number $k$ while a $\Notify$ summary may not.
	If multiple summaries claim the same highest iteration number, the $\Commit$ certificate can be for any of them.  
	Without loss of generality, suppose $P = (s_1, s_2, \cdots, s_{f+1}, \Cert)$ in which $s_j = \sig{k, \Status, v_j, k_j}_j$ or $s_j = \sig{k_j, \Notify, v_j}$. 
	Let $k^*=\max(k_1,k_2,\cdots,k_{f+1})$.
	If $k^* = 0$, then $\Cert=\bot$ and any $v$ is safe to propose under $P$.
	If $k^* > 0$, then $v$ must match some $v_j$ such that $k_j = k^*$ and $\Cert$ certifies $(v_j,k_j)$.
	There may be additional requirements on $v$ at the application level that replicas should verify.
	\\
\hline

$\Cert$ certifies $(v,k)$
	& 
	$\Cert$ consists of $f+1$ $\Commit$ requests or $\Notify$ summaries for value $v$.
	Each $\Commit$ request in $\Cert$ must carry the current iteration number $k$ while a $\Notify$ summary may not.
	\\
\hline
\end{tabular}
\end{table}

\subsection{Safety and Liveness}	\label{sec:proof}

In this section, we prove that the protocol in Section~\ref{sec:main-protocol} provides safety and liveness.
We will also give intuition to aid understanding. 

The scenario to consider for safety is when an honest replica $h$ commits on a value $v^*$ in iteration $k^*$. 
We show that, in all subsequent iterations, no leader can construct a valid proposal for a value other than $v^*$.
We first show that Byzantine replicas cannot commit or accept a value other than $v^*$ in iteration $k^*$. 
Thus, all other honest replicas accept $v^*$ at the end of iteration $k^*$ upon receiving $\Notify$ from the honest replica $h$.
The leader in iteration $k^*+1$ needs to show $f+1$ $\Status$ messages and pick a value with the highest iteration number (cf. Table~\ref{tab:valid}). 
One of these $\Status$ messages must be from an honest replica and contain $v^*$.
This implies that a value other than $v^*$ cannot be proposed in iteration $k^*+1$, and hence cannot be committed or accepted in iteration $k^*+1$, and hence cannot be proposed in iteration $k^*+2$, and so on.
Safety then holds by induction.

We now formalize the above intuition through an analysis on $\Commit$ certificates. 
A $\Commit$ certificate $\Cert$ certifies that $v$ has been legally committed and/or accepted in iteration $k$, if it meets the validity requirement in Table~\ref{tab:valid}.
We prove the following lemma about $\Commit$ certificates: once an honest replica commits, all $\Commit$ certificates in that iteration and future iterations can only certify its committed value.

\begin{lemma}
Suppose replica $h$ is the first honest replica to commit.
If replica $h$ commits on $v^*$ in iteration $k^*$ and $\Cert$ certifies $(v,k)$ where $k \geq k^*$, then $v=v^*$.
\label{lemma:cert}
\end{lemma}
\begin{proof}
We prove by induction on $k$.
For the base case, suppose $\Cert$ certifies $(v,k^*)$.
$\Cert$ must consist of $f+1$ valid $\Commit$ requests or $\Notify$ summaries for $v$.
At least one of these comes from an honest replica (call it $h_1$).
Since no honest replica has terminated so far,  
replica $h_1$ must have sent a normal $\Commit$ request rather than a virtual one.
Thus, replica $h_1$ must have received a valid proposal (could be virtual) for $v$ from the leader, 
and must have forwarded the proposal to all other replicas.
If $v \neq v^*$, replica $h$ would have detected leader equivocation, and would not have committed on $v^*$ in this iteration.
So we have $v=v^*$.

Before proceeding to the inductive case, it is important to observe that all honest replicas will accept $v^*$ at the end of iteration $k^*$.
This is because, in iteration $k^*$, the honest replica $h$ must have sent a $\Notify$ message (with a $\Commit$ certificate for $v^*$) to all replicas and $\Commit$ certificates for other values cannot exist.
Now for the inductive case, suppose the lemma holds up to iteration $k$.
We need to prove that if $\Cert$ certifies $(v, k+1)$, then $v = v^*$.
The inductive hypothesis says between iteration $k^*$ to iteration $k$, all $\Commit$ certificates certify $v^*$. 
So at the beginning of iteration $k+1$, all honest replicas either have committed on $v^*$ or still accept $v^*$. 
If $\Cert$ certifies $(v,k+1)$, it must consist of $f+1$ valid $\Commit$ requests or $\Notify$ summaries for $v$. 
At least one of these is from an honest replica (call it $h_2$).
\begin{enumerate}
\item 
If replica $h_2$ has terminated before iteration $k+1$, its $\Notify$ summary (virtual $\Commit$ request) is for $v^*$ and we have $v=v^*$. 

\item
Otherwise, replica $h_2$ must have received from the leader a valid proposal (could be virtual) for $v$ in iteration $k+1$.
Note again that all honest replicas either have committed on $v^*$ or still accept $v^*$ at the beginning of iteration $k+1$. 
\begin{enumerate}
\item
If the proposal is a virtual one, 
in order for replica $h_2$ to consider it valid, 
$v$ must match $v_{h_2}$ (the value replica $h_2$ accepts), which is $v^*$.

\item
If the proposal is a normal one, then it must contain a safe value proof $P$ for $v$. 
$P$ must include at least one honest replica's $\Status$ or $\Notify$ summary for $(v^*,k^*)$ (or an even higher iteration number). 
Due to the inductive hypothesis, $P$ cannot contain a $\Commit$ certificate for $(v' ,k')$ where $k' \geq k^*$ and $v' \neq v^*$. 
Recall that $v$ must match the value in a summary that claims the highest iteration number.
Therefore, the only value that is safe to propose under $P$ is $v=v^*$.
\end{enumerate}
\end{enumerate} 
Therefore, we have $v=v'$ in all cases in the inductive step, completing the proof.
\end{proof}

\begin{theorem}[Safety]
If two honest replicas commit on $v$ and $v'$ respectively, then $v=v'$.
\label{thm:safety}
\end{theorem}
\begin{proof}
Suppose replica $h$ is the first honest replica to commit, and it commits on $v^*$ in iteration $k^*$.
In order for another honest replica to commit on $v$, there must be a valid $\Commit$ certificate $\Cert$ for $(v,k)$ where $k \geq k^*$.  
Due to Lemma~\ref{lemma:cert}, $v=v^*$.
Similarly, $v'=v^*$, and we have $v=v'$.
\end{proof}

Now we move on to liveness and show that an honest leader will guarantee that all honest replicas terminate by the end of that iteration. 

\begin{theorem}[Liveness]
If the leader $L$ in iteration $k$ is honest, then every honest replica terminates at the end of iteration $k$ (if it has not already terminated before iteration $k$). 
\end{theorem}
\begin{proof}
The honest leader $L$ will send a proposal (could be virtual) to all replicas.
If $L$ has not terminated, it will send a valid proposal for a value $v$ that is safe to propose, and attach a valid proof $P$.
If $L$ has committed on $v$ and terminated, then all honest replicas have either committed on $v$ or accept $v$ at the beginning of iteration $k$ (see proof of Lemma~\ref{lemma:cert}), 
so they all consider $L$'s virtual proposal valid.
Additionally, the unforgeability of digital signatures prevents Byzantine replicas from falsely accusing $L$ of equivocating.
Therefore, all honest replicas (terminated or otherwise) will send $\Commit$ requests (could be virtual) for $v$, receive $f+1$ $\Commit$ requests for $v$ and terminate at the end of the iteration.
\end{proof}

Finally, we mention an interesting scenario that does not have to be explicitly addressed in the proofs.
Before any honest replica commits, Byzantine replicas may obtain $\Commit$ certificates for multiple different values in the same iteration. 
In particular, the Byzantine leader proposes two values $v$ and $v'$ to all the $f$ Byzantine replicas.
(An example with more than two values is similar.) 
Byzantine replicas then exchange $f$ $\Commit$ requests for both values among them.
Additionally, the Byzantine leader proposes $v$ and $v'$ to different honest replicas.
Now with one more $\Commit$ request for each value from honest replicas,
Byzantine replicas can obtain $\Commit$ certificates for both $v$ and $v'$, 
and can make honest replicas accept different values by showing them different $\Commit$ certificates ($\Notify$ messages).
However, this will not lead to a safety violation because no honest replica would have committed in this iteration:
the leader has equivocated to honest replicas, so all honest replicas will detect equivocation from forwarded proposals and thus refuse to commit.
This scenario showcases the necessity of both the synchrony assumption and the use of digital signatures for our protocol. 
Lacking either one, equivocation cannot be reliably detected and any protocol will be subject to the $f<n/3$ bound.
For completeness, we note that the above scenario will not lead to a liveness violation, either.
In the next iteration, honest replicas consider a proposal for any value (including but not limited to $v$ and $v'$) to be valid as long as it contains a valid safe value proof $P$ for that value.


\section{Byzantine Fault Tolerant State Machine Replication}	\label{sec:smr}

\subsection{Model and Overview}

The state machine replication approach for fault tolerance considers a scenario where clients submit requests to a replicated service~\cite{Lamport78,Schneider90}.
The replicated service should provide safety and liveness even when some replicas are faulty ($f<n/2$ Byzantine faults in our case).
Safety means the service behaves like a single non-faulty server, and liveness means the service keeps processing client requests and eventually commits every request~\cite{Schneider90,PBFT}. 
To satisfy safety, honest replicas should agree on a sequence of values and their order. 
We say each value in the sequence occupies a \emph{slot}.

We remark that BFT state machine replication requires an honest majority (i.e., $n \geq 2f+1$) even in the synchronous setting~\cite{Schneider90}.
Otherwise, Byzantine replicas in the majority can convince clients of a decision and later deny having ever committed that decision.
This is in contrast to Byzantine broadcast which can be solved even for $n/2 \leq f \leq n-2$.
In Byzantine broadcast, honest parties just need to stay in agreement with each other and do not have to convince external entities (e.g., clients) of the correct system states.
This distinction between BFT state machine replication and Byzantine broadcast becomes unimportant in the asynchronous setting in which both problems require $n \geq 3f+1$.

The rest of this section presents two state machine replication protocols.
We start with a basic protocol, which extends the synod protocol with minimum modifications and requires amortized 2 iterations (8 rounds) per slot.
The second protocol improves the common case (with a stable leader) at the cost of more expensive view (leader) changes.
It achieves 1 iteration (3 rounds) per slot in the common case.

\subsection{A Basic Protocol}	\label{sec:smr:basic}
The basic protocol essentially runs a series of synod instances sequentially.
To start, the following modifications are natural. 
\begin{enumerate}
\item Each replica $i$ internally maintains an additional long-term state $s_i$ to denote which slot it is currently working on. $s_i=s$ means replica $i$ has committed for slots $\{1,2,\cdots,s-1\}$, and has not committed for slots $s, s+1, s+2, \cdots$.

\item All messages in Section~\ref{sec:main-protocol} and Table~\ref{tab:valid} contain an additional slot number $s$.
The four types of messages now have the form: 
$\sig{\sig{s, k, \Status, v_i, k_i}_i, \Cert_i}_i$,
$\sig{\sig{s, k, \Propose, v}_L, P}_L$,
$\sig{s, k, \Commit, v}_i$, and
$\sig{\sig{s, k, \Notify, v}_i, \Cert}_i$. 
A $\Commit$ certificate $\Cert$ now certifies a triplet $(s,v,k)$ and must contain $f+1$ $\Commit$ requests (could be virtual) for the same slot $s$.
Similarly, a safe value proof $P$ must now contain $f+1$ $\Status$ summaries (could be virtual) for the same slot $s$.
\end{enumerate}

Replica $i$ follows the protocol in Section~\ref{sec:main-protocol} to send and react to messages for the slot $s_i$ it is currently working on. 
If replica $i$ receives messages for a past slot $s<s_i$, it can safely ignore them because its virtual messages ($\Notify$ messages) for slot $s$ will suffice to help any honest replicas terminate for that slot.  
The more interesting question is how to react to messages for future slots. 
We observe that it is vital that all honest replicas accept values for future slots upon receiving valid $\Notify$ messages.
If the sender of the $\Notify$ message is an honest replica, then all other honest replicas must accept that value to prevent any other value from being proposed for that slot. 
However, apart from $\Notify$ messages, honest replicas can ignore all other messages for future slots without violating safety.
The proof for safety from Section~\ref{sec:proof} still holds since it only relies on honest replicas \emph{not sending improper messages}.
This implies the following changes:
\begin{enumerate}[resume]
\item The accepted states of a replica $i$ are now per slot: $\Accepted_i[s] = (v_i[s], k_i[s], \Cert_i[s])$. 
Upon receiving a valid $\Notify$ message $\sig{\sig{s, k, \Notify, v}_j, \Cert}_j$ for a future slot $s>s_i$, 
replica $i$ sets $\Accepted_i[s] := (v, k, \Cert)$.
Replica $i$ ignores messages other than $\Notify$ for future slots.
\end{enumerate}

We next analyze liveness and the amortized round complexity for our state machine replication protocol.
Note that now an honest leader may not be able to ensure termination for the slot it is working on because some honest replicas may be lagging behind.
They may not react to its proposal; in fact, the leader may not even gather enough $\Status$ messages to construct a valid proposal.
However, each honest leader can still guarantee termination for at least one slot: the lowest numbered slot $s^*$ that any honest replica is working on.
After this iteration, all honest replicas will at least be working on slot $s^*+1$.
Therefore, we simply rotate the leader in a round robin fashion.
This way, we can fill at least $f+1$ slots in a full rotation of $2f+1$ iterations, thereby achieving amortized 2 iterations (8 rounds) per slot.

\paragraph{Reply to clients.}
In our protocol, a single replica's $\Notify$ message is insufficient to convince a client that a value is committed.
The reason is that a Byzantine replica may not send $\Notify$ message to all other replicas, in which case a different value may be committed later.
Instead, a client needs to see $f+1$ $\Notify$ messages (summaries suffice) from distinct replicas to be confident of a committed value.
When $n$ is large, it may be costly for a client to maintain connections to all $n$ replicas. 
An alternative solution is to let a replica gather $f+1$ $\Notify$ summaries at the end of Round 3, and send them to a client in a single message.  
We call these $f+1$ $\Notify$ summaries a $\Notify$ certificate $\NotSum$.
A client can be assured of a committed value by receiving $\NotSum$ from a single replica at the cost of one extra round of latency.

\subsection{Towards Stable Leaders}
The previous protocol replaces the leader after every iteration.
This is not ideal since faulty leaders may prevent progress in their iterations.
As a result, the previous protocol may only make progress every one out of two iterations on average. 
A better design, which is common in PBFT-style protocols, is to keep a stable leader in charge and replace the leader only if it is detected to be faulty. 
This way, once an honest leader is in control, a slot can be committed after every iteration. 

Following PBFT's terminology, we say a protocol proceeds in a series of \emph{views}. 
Each view has a unique leader and we again rotate the leader in a round robin fashion.
A view can last many iterations, potentially indefinitely as long as the leader keeps making progress.
Suppose the current view is view $l$ and its leader is $L_l$. 
(Note that $L_l$ has been redefined to be the leader of a view as against an iteration.)
At a high level, if an honest replica $i$ detects that $L_l$ is faulty, 
it broadcasts a $\sig{\Blame, l+1}_i$ message,
which can be thought of as an accusation against $L_l$. 
If the next leader $L_{l+1}$ gathers $f+1$ $\sig{\Blame, l+1}_i$ messages from distinct replicas, 
it broadcasts a $\NewView$ message to become the new leader and the system enters view $l+1$.
Though the above idea seems simple at a high level, the details involved are quite intricate. 
Due to lack of space, we give an overview of the key challenges and our solutions here and present the detailed protocol and the proof in Appendix~\ref{sec:smr:stable_leader}.

First note that a faulty leader can always cause honest replicas to disagree on whether or not it is faulty.
To do so, the faulty leader and all Byzantine replicas just behave normally to some honest replicas (call them group 1) and remain silent to other honest replicas (call them group 2).
In this case, replicas in group 2 will accuse the leader, but they cannot convince group 1 because from group 1's perspective, it is entirely possible that replicas in group 2 are malicious and are falsely accusing an innocent leader.
Fortunately, we can still ensure progress by utilizing the following property of our synod protocol:
If $L_l$ is honest, then a replica expects to not only to commit but also receive at least $f+1$ $\Notify$ messages from distinct replicas (including itself) at the end of Round 3 of each iteration.
These $f+1$ $\Notify$ messages form a $\Notify$ certificate $\NotSum$ and convince any other replica to commit (cf. Section~\ref{sec:smr:basic}).
If replica $i$ does not receive $f+1$ $\Notify$ messages at the end of Round 3, then it knows the current leader is faulty and broadcasts $\Blame$.
Therefore, after each iteration, either some honest replica obtains the ability to convince other replicas to commit, or all honest replica accuse the current leader and the next leader can start a new view. 

However, to complicate the matter, the next leader may also be Byzantine.
It may not send $\NewView$ when it is supposed to, or send it only to a subset of honest replicas.
We need to ensure such behaviors do not violate safety or liveness.
For safety, we would like to ensure that two honest replicas are never in different views. 
Unfortunately, we do not see a way to keep all honest replicas always in the same view.
Instead, our protocol guarantees the following: if an honest replica enters view $l+1$, then all other honest replicas exit view $l$ --- they may not enter view $l+1$, which means they temporarily may not be in any view.
To guarantee liveness, these ``out-of-view'' replicas must eventually enter a future view.  
This is achieved by ensuring that, if $L_{l+1}$ does not send $\NewView$ or tries to prevent progress in any other way, then all honest replicas will accuse $L_{l+1}$. 
Meanwhile, we also need to ensure that an honest replica will not be tricked into accusing an honest future leader.
Thus, before accusing $L_{l+1}$, an honest replica $i$ sends the $f+1$ $\Blame$ messages (a $\Blame$ certificate) to $L_{l+1}$;
if $L_{l+1}$ still does not step up in the round after, then replica $i$ can be certain that $L_{l+1}$ is faulty.

\paragraph{Message complexity and digital signatures.}
The $\Status$ round and the full $\Notify$ round can be pushed to the view change procedure.
With this change, our improved protocol will have very similar complexity as the Practical Byzantine Fault Tolerance (PBFT) protocol.
Like PBFT, in the common case (i.e., under a stable leader), our protocol has three rounds, two of which require all-to-all communication of $O(1)$-sized messages. 
$\Theta(n)$-sized messages are needed only for view changes.
System-level optimizations like parallel slots, checkpoints and garbage collection~\cite{PBFT} can also be added. 
It is worth noting that we require digital signatures even in the common case, whereas PBFT uses digital signatures only for view changes. 
This was a major contribution of PBFT two decades ago as digital signatures were very slow back then. 
But with computation becoming cheaper and the development of more efficient signature schemes~\cite{ECDSA}, 
we believe the use of signatures is less of a concern today. 
Perhaps a more important question is whether the synchrony assumption itself is practical.
The next subsection discusses this topic.

\subsection{Clock Synchronization}

The synchrony assumption essentially states that all honest replicas' messages arrive in time.
This requires two properties:
(i) a bounded message delay and 
(ii) locked step execution, i.e., honest replicas enter each round roughly at the same time. 
The second property is important because, if replica $i$ enters a round much earlier than replica $j$, then $i$ may end up finishing the round too soon without waiting for $j$'s message to arrive.
In our protocol, for example, this could prevent $i$ from detecting leader equivocation and result in a safety violation.  

The XFT paper provided some justification for the bounded message delay assumption in certain applications~\cite{XFT}. 
But we still a mechanism to enforce locked step execution.
To this end, we will use the following clock synchronization protocol.
It will be executed at known time intervals.
We call each interval a ``day''. 
\begin{itemize}
\addtolength{\itemindent}{32pt}
\item[\textbf{Round 0}] ($\Sync$)
When replica $i$'s clock reaches the beginning of day $X$, it sends a $\sig{\Sync, X}_i$ message to all replicas including itself.
\item[\textbf{Round 1}] ($\SetTime$) 
The first time a replica $j$ receives $f+1$ $\sig{\Sync, X}$ messages from distinct replicas (either as $f+1$ separate $\Sync$ messages or within a single $\SetTime$ message), it
\begin{itemize}
\item sets its clock to the beginning of day $X$, and
\item sends all other replicas a $\SetTime$ message, which is the concatenation of $f+1$ $\sig{\Sync,X}$ messages from distinct replicas.
\end{itemize}
\end{itemize}

The above protocol refreshes honest replicas' clock difference to at most the message delay bound $\delta$ at the beginning of each day. 
The first honest replica to start a new day will broadcast a $\SetTime$ message, which makes all other honest replicas start the new day within $\delta$ time.
Obtaining a $\SetTime$ message also means at least one honest replica has sent a valid $\Sync$ message, ensuring that roughly one day has indeed passed since the previous day. 
We can then set the duration of each round to $2\delta+\drift$ where $\drift$ is the maximum clock drift between two honest replicas in a ``day''.

This clock synchronization protocol may be of independent interest to synchronous protocols other than Byzantine consensus.
We also note that it does not require a locked step execution.
Each $\Sync$ message is triggered by a replica's own local clock, independent of when day $X$ would start for other replicas. 

\paragraph{Best-case optimization.}
A replica does not need to send a $\SetTime$ message, if
(i) it receives $2f+1$ distinct $\sig{\Sync, X}$ messages (as separate $\Sync$ messages, via $\SetTime$ messages or a mixture of both), or (ii) it receives $f+1$ distinct $\SetTime$ messages.
In scenario (i), all honest replicas have sent a $\Sync$ message.
In scenario (ii), some honest replica has sent $\SetTime$.
So the replica can be assured that all other honest replicas will enter the new day within $\delta$ time.


\section{Byzantine Broadcast and Agreement}	\label{sec:BABB}

\subsection{Byzantine Broadcast}

In Byzantine broadcast, there is a designated \emph{sender} who tries to broadcast a value to $n$ parties.
A solution needs to satisfy three requirements:
\begin{enumerate}[itemsep=0pt]
\item[] \textbf{(termination)} all honest parties eventually commit,
\item[] \textbf{(agreement)} all honest parties commit on the same value, and
\item[] \textbf{(validity)} if the sender is honest, then all honest parties commit on the value it broadcasts.
\end{enumerate} 
In this section, we describe a protocol that solves synchronous authenticated Byzantine broadcast for the $f<n/2$ case.

\paragraph{A ``pre-round''.}
Our core protocol in Section~\ref{sec:main-protocol} can be used to satisfy the agreement and termination requirement. 
But to satisfy validity, we need to prepend an extra round to allow the designated sender to broadcast its value.
We will call this the \emph{pre-round}.
Let replica $L_s$ be the designated sender.  
In the pre-round, $L_s$ broadcasts a signed value $\sig{v_s}_{L_s}$ to every replica.
At the end of the pre-round, if replica $i$ receives $\sig{v_s}_{L_s}$ from $L_s$, then replica $i$ accepts $v_s$ by setting its states $(v_i, k_i) := (v_s, 0)$ and $\Cert_i := \sig{v_s}_{L_s}$.
Note that the certificate for accepting $v_s$ is simply a valid signature from $L_s$.

\paragraph{Main loop.}
After the pre-round, we enter the main loop of the core protocol in Section~\ref{sec:main-protocol}. 
The first iteration of the main loop is now slightly different. 
Before, the leader $L_1$ of the first iteration will certainly find in its $\Status$ round that ``no replica has accepted anything'', and is hence free to propose any value.
But now, replicas may have accepted values from $L_s$.

\paragraph{Agreement and Validity.}
If the designated sender $L_s$ is honest, then all honest replicas will accept its value before entering the main loop. 
Thus, any safe value proof $P$ that $L_1$ can construct will contain at least one $\Status$ message vouching for $v_s$.
Following an inductive proof similar to Theorem~\ref{thm:safety}, no value other than $v_s$ can have a certificate or be proposed in any iteration, satisfying validity.
The first honest leader in the main loop will ensure agreement and termination.
If $L_s$ is faulty, it may send no value or multiple values in the pre-round.
As before, in the $\Propose$ round, if the leader finds a tie for the highest iteration number (define ``pre-round'' to have iteration number 0) between multiple legally accepted values, it is free to pick any of them.
This means if $L_s$ sends multiple values (or no value), it is up to future leaders to pick (or propose) a value for the replicas to agree on.
Agreement and termination are guaranteed by the core protocol.

\paragraph{Efficiency.}
Assuming we have a random leader oracle, there is a $(f+1)/(2f+1) > 1/2$ probability that each leader after the first iteration is honest, so the protocol terminates in expected 2 iterations after the pre-round. 
Since honest replicas terminate without waiting for the $\Notify$ messages in the last iteration,
the protocol terminates in $1 + 2 \times 4 - 1 = 8$ rounds in expectation. 

\paragraph{Random leader oracle.}
The remaining question is how to elect random leaders.
For this step, we can adopt the moderated verifiable secret sharing approach from Katz and Koo~\cite{KK06} or the unique signature approach from Algorand~\cite{Algorand}.
The resulting protocol does not yet achieve expected constant rounds against an adaptive adversary who can corrupt a leader as soon as it is elected.
Inspired by prior work~\cite{KK06,Algorand}, the solution is to elect a random leader only \emph{after} it has fulfilled the leader's responsibility.
Adapting the idea to our protocol, every party should act as a leader in the $\Status$ round and the $\Propose$ round of each iteration to collect states and make a proposal.
A random leader is then elected after the $\Propose$ round and before the $\Commit$ round.
This ensures a $1/2$ chance of having an honest leader in each iteration even against an adaptive adversary. 

\subsection{Byzantine Agreement}
In Byzantine agreement, every party holds an initial input value.
A solution needs to satisfy the same termination and agreement requirements as in Byzantine broadcast.
There exist a few different validity notions.
We adopt a common one known as strong unanimity~\cite{DLS88}:
\begin{enumerate}[itemsep=0pt]
\item[] \textbf{(validity)} if all honest parties hold the same input value $v$, then they all commit on $v$.
\end{enumerate} 
To solve Byzantine agreement, we can use a classical transformation~\cite{PSL80,LSP82}.
Each party initiates a Byzantine broadcast in parallel to broadcast its value. 
After the broadcast, every honest party will share the same vector of values {$V=\{v_1, v_2, \cdots, v_n\}$}.
If party $i$ is honest, then $v_i$ will be its input value.
Each party can then just output the most frequent value in the vector $V$.
Agreement is reached since all honest parties share the same $V$.
If all honest parties start with the same input value, then that value will be the most frequent in $V$, achieving validity. 
After the first iteration, we can elect a single leader for all the parallel broadcast instances.
Thus, the agreement protocol will have the same round complexity as the broadcast protocol.

\section{Conclusion and Future Work}
This paper has described a 4-round synod protocol that tolerates $f$ Byzantine faults using $2f+1$ replicas in the synchronous and authenticated setting.
We then apply the synod protocol to achieve BFT state machine replication in amortized 3 rounds per slot and solve Byzantine agreement in expected 8 rounds.
Our protocols may be applied to build synchronous BFT systems or improve cryptographic protocols. 

\paragraph{Acknowledgements.}
The authors thank Dahlia Malkhi for helpful discussions.
This work is funded in part by NSF award \#1518765, a Google Ph.D. Fellowship award, the HUJI Cyber Security Research Center and the Israel National Cyber Bureau in the Prime Minister's Office.

\bibliographystyle{plainurl}
\bibliography{refs}

\begin{thebibliography}{10}

\bibitem{Abd05}
Michael Abd-El-Malek, Gregory~R Ganger, Garth~R Goodson, Michael~K Reiter, and
  Jay~J Wylie.
\newblock Fault-scalable byzantine fault-tolerant services.
\newblock In {\em ACM SIGOPS Operating Systems Review}, volume~39, pages
  59--74. ACM, 2005.

\bibitem{Solidus}
Ittai Abraham, Dahlia Malkhi, Kartik Nayak, Ling Ren, and Alexander Spiegelman.
\newblock Solidus: An incentive-compatible cryptocurrency based on
  permissionless byzantine consensus.
\newblock arXiv preprint, 2016.

\bibitem{Farsite}
Atul Adya, William~J Bolosky, Miguel Castro, Gerald Cermak, Ronnie Chaiken,
  John~R Douceur, Jon Howell, Jacob~R Lorch, Marvin Theimer, and Roger~P
  Wattenhofer.
\newblock {FARSITE}: Federated, available, and reliable storage for an
  incompletely trusted environment.
\newblock {\em ACM SIGOPS Operating Systems Review}, 36(SI):1--14, 2002.

\bibitem{BenOr83}
Michael Ben-Or.
\newblock Another advantage of free choice (extended abstract): Completely
  asynchronous agreement protocols.
\newblock In {\em Proceedings of the second annual ACM symposium on Principles
  of distributed computing}, pages 27--30. ACM, 1983.

\bibitem{BGW88}
Michael Ben-Or, Shafi Goldwasser, and Avi Wigderson.
\newblock Completeness theorems for non-cryptographic fault-tolerant
  distributed computation.
\newblock In {\em Proceedings of the twentieth annual ACM symposium on Theory
  of computing}, pages 1--10. ACM, 1988.

\bibitem{Bracha87}
Gabriel Bracha.
\newblock Asynchronous byzantine agreement protocols.
\newblock {\em Information and Computation}, 75(2):130--143, 1987.

\bibitem{PBFT}
Miguel Castro and Barbara Liskov.
\newblock Practical byzantine fault tolerance.
\newblock In {\em Proceedings of the third symposium on Operating systems
  design and implementation}, pages 173--186. USENIX Association, 1999.

\bibitem{PBFT02}
Miguel Castro and Barbara Liskov.
\newblock Practical byzantine fault tolerance and proactive recovery.
\newblock {\em ACM Transactions on Computer Systems}, 20(4):398--461, 2002.

\bibitem{Chun07}
Byung-Gon Chun, Petros Maniatis, Scott Shenker, and John Kubiatowicz.
\newblock Attested append-only memory: Making adversaries stick to their word.
\newblock In {\em ACM SIGOPS Operating Systems Review}, volume~41, pages
  189--204. ACM, 2007.

\bibitem{Correia04}
Miguel Correia, Nuno~Ferreira Neves, and Paulo Verissimo.
\newblock How to tolerate half less one byzantine nodes in practical
  distributed systems.
\newblock In {\em Reliable Distributed Systems, 2004. Proceedings of the 23rd
  IEEE International Symposium on}, pages 174--183. IEEE, 2004.

\bibitem{HQRep06}
James Cowling, Daniel Myers, Barbara Liskov, Rodrigo Rodrigues, and Liuba
  Shrira.
\newblock {HQ} replication: A hybrid quorum protocol for byzantine fault
  tolerance.
\newblock In {\em 7th symposium on Operating systems design and
  implementation}, pages 177--190. USENIX Association, 2006.

\bibitem{DS83}
Danny Dolev and H.~Raymond Strong.
\newblock Authenticated algorithms for byzantine agreement.
\newblock {\em SIAM Journal on Computing}, 12(4):656--666, 1983.

\bibitem{DLS88}
Cynthia Dwork, Nancy Lynch, and Larry Stockmeyer.
\newblock Consensus in the presence of partial synchrony.
\newblock {\em Journal of the ACM}, 35(2):288--323, 1988.

\bibitem{FM97}
Pesech Feldman and Silvio Micali.
\newblock An optimal probabilistic protocol for synchronous byzantine
  agreement.
\newblock {\em SIAM Journal on Computing}, 26(4):873--933, 1997.

\bibitem{FL82}
Michael~J. Fischer and Nancy~A. Lynch.
\newblock A lower bound for the time to assure interactive consistency.
\newblock {\em Information processing letters}, 14(4):183--186, 1982.

\bibitem{FLP85}
Michael~J. Fischer, Nancy~A. Lynch, and Michael~S. Paterson.
\newblock Impossibility of distributed consensus with one faulty process.
\newblock {\em Journal of the ACM}, 32(2):374--382, 1985.

\bibitem{Garay98}
Juan~A Garay and Yoram Moses.
\newblock Fully polynomial byzantine agreement for $n > 3f$ processors in $f+1$
  rounds.
\newblock {\em SIAM Journal on Computing}, 27(1):247--290, 1998.

\bibitem{GMW87}
Shafi Goldwasser, Silvio Micali, and Avi Wigderson.
\newblock How to play any mental game, or a completeness theorem for protocols
  with an honest majority.
\newblock In {\em Proc. of the 19th Annual ACM STOC}, volume~87, pages
  218--229, 1987.

\bibitem{ECDSA}
Don Johnson, Alfred Menezes, and Scott Vanstone.
\newblock The elliptic curve digital signature algorithm (ecdsa).
\newblock {\em International Journal of Information Security}, 1(1):36--63,
  2001.

\bibitem{CheapBFT12}
R{\"u}diger Kapitza, Johannes Behl, Christian Cachin, Tobias Distler, Simon
  Kuhnle, Seyed~Vahid Mohammadi, Wolfgang Schr{\"o}der-Preikschat, and Klaus
  Stengel.
\newblock Cheapbft: resource-efficient byzantine fault tolerance.
\newblock In {\em Proceedings of the 7th ACM european conference on Computer
  Systems}, pages 295--308. ACM, 2012.

\bibitem{KK06}
Jonathan Katz and Chiu-Yuen Koo.
\newblock On expected constant-round protocols for byzantine agreement.
\newblock {\em J. Comput. Syst. Sci.}, 75(2):91--112, 2009.

\bibitem{Byzcoin}
Eleftherios~Kokoris Kogias, Philipp Jovanovic, Nicolas Gailly, Ismail Khoffi,
  Linus Gasser, and Bryan Ford.
\newblock Enhancing bitcoin security and performance with strong consistency
  via collective signing.
\newblock In {\em 25th USENIX Security Symposium}, pages 279--296. USENIX
  Association, 2016.

\bibitem{Zyzzyva}
Ramakrishna Kotla, Lorenzo Alvisi, Mike Dahlin, Allen Clement, and Edmund Wong.
\newblock Zyzzyva: speculative byzantine fault tolerance.
\newblock In {\em ACM SIGOPS Operating Systems Review}, volume~41, pages
  45--58. ACM, 2007.

\bibitem{Kotla04}
Ramakrishna Kotla and Mike Dahlin.
\newblock High throughput byzantine fault tolerance.
\newblock In {\em Dependable Systems and Networks}, pages 575--584. IEEE, 2004.

\bibitem{Oceanstore}
John Kubiatowicz, David Bindel, Yan Chen, Steven Czerwinski, Patrick Eaton,
  Dennis Geels, Ramakrishan Gummadi, Sean Rhea, Hakim Weatherspoon, Westley
  Weimer, and Ben Zhao.
\newblock Oceanstore: An architecture for global-scale persistent storage.
\newblock {\em ACM Sigplan Notices}, 35(11):190--201, 2000.

\bibitem{Lamport78}
Leslie Lamport.
\newblock Time, clocks, and the ordering of events in a distributed system.
\newblock {\em Communications of the ACM}, 21(7):558--565, 1978.

\bibitem{Paxos}
Leslie Lamport.
\newblock The part-time parliament.
\newblock {\em ACM Transactions on Computer Systems}, 16(2):133--169, 1998.

\bibitem{LSP82}
Leslie Lamport, Robert Shostak, and Marshall Pease.
\newblock The byzantine generals problem.
\newblock {\em ACM Transactions on Programming Languages and Systems},
  4(3):382--401, 1982.

\bibitem{Lindell02}
Yehuda Lindell, Anna Lysyanskaya, and Tal Rabin.
\newblock Sequential composition of protocols without simultaneous termination.
\newblock In {\em Proceedings of the twenty-first annual symposium on
  Principles of distributed computing}, pages 203--212. ACM, 2002.

\bibitem{XFT}
Shengyun Liu, Christian Cachin, Vivien Qu{\'e}ma, and Marko Vukolic.
\newblock {XFT}: practical fault tolerance beyond crashes.
\newblock In {\em 12th USENIX Symposium on Operating Systems Design and
  Implementation}, pages 485--500. USENIX Association, 2016.

\bibitem{Martin06}
J-P Martin and Lorenzo Alvisi.
\newblock Fast byzantine consensus.
\newblock {\em IEEE Transactions on Dependable and Secure Computing},
  3(3):202--215, 2006.

\bibitem{Algorand}
Silvio Micali.
\newblock Algorand: The efficient and democratic ledger.
\newblock arXiv preprint, 2016.

\bibitem{Viewstamped}
Brian~M. Oki and Barbara~H. Liskov.
\newblock Viewstamped replication: A new primary copy method to support
  highly-available distributed systems.
\newblock In {\em Proceedings of the seventh annual ACM Symposium on Principles
  of distributed computing}, pages 8--17. ACM, 1988.

\bibitem{HybridConsensus}
Rafael Pass and Elaine Shi.
\newblock Hybrid consensus: Efficient consensus in the permissionless model.
\newblock Cryptology ePrint Archive, Report 2016/917, 2016.

\bibitem{PSL80}
Marshall Pease, Robert Shostak, and Leslie Lamport.
\newblock Reaching agreement in the presence of faults.
\newblock {\em Journal of the ACM}, 27(2):228--234, 1980.

\bibitem{Rabin83}
Michael~O. Rabin.
\newblock Randomized byzantine generals.
\newblock In {\em Proceedings of the 24th Annual Symposium on Foundations of
  Computer Science}, pages 403--409. IEEE, 1983.

\bibitem{Omega96}
Michael~K. Reiter, Matthew~K. Franklin, John~B. Lacy, and Rebecca~N. Wright.
\newblock The {$\Omega$} key management service.
\newblock In {\em Proceedings of the 3rd ACM conference on Computer and
  communications security}, pages 38--47. ACM, 1996.

\bibitem{BASE01}
Rodrigo Rodrigues, Miguel Castro, and Barbara Liskov.
\newblock {BASE}: Using abstraction to improve fault tolerance.
\newblock In {\em ACM SIGOPS Operating Systems Review}, volume~35, pages
  15--28. ACM, 2001.

\bibitem{Schneider90}
Fred~B Schneider.
\newblock Implementing fault-tolerant services using the state machine
  approach: A tutorial.
\newblock {\em ACM Computing Surveys}, 22(4):299--319, 1990.

\bibitem{Veronese13}
Giuliana~Santos Veronese, Miguel Correia, Alysson~Neves Bessani, Lau~Cheuk
  Lung, and Paulo Verissimo.
\newblock Efficient byzantine fault-tolerance.
\newblock {\em IEEE Transactions on Computers}, 62(1):16--30, 2013.

\bibitem{Wallach03}
Dan~S Wallach, Peter Druschel, et~al.
\newblock Enforcing fair sharing of peer-to-peer resources.
\newblock In {\em International Workshop on Peer-to-Peer Systems}, pages
  149--159. Springer, 2003.

\bibitem{Yin03}
Jian Yin, Jean-Philippe Martin, Arun Venkataramani, Lorenzo Alvisi, and Mike
  Dahlin.
\newblock Separating agreement from execution for byzantine fault tolerant
  services.
\newblock {\em ACM SIGOPS Operating Systems Review}, 37(5):253--267, 2003.

\bibitem{COCA02}
Lidong Zhou, Fred Schneider, and Robbert van Renesse.
\newblock {COCA}: A secure distributed online certification authority.
\newblock {\em ACM Transactions on Computer Systems}, 20(4):329--368, 2002.

\end{thebibliography}

\newpage
\appendix

\section{An Illustration of the Core Protocol}
Figure~\ref{fig:protocol} gives an illustration of the core synod protocol in Section~\ref{sec:main-protocol}.

\begin{figure}[bt]
\centering
\includegraphics[width=0.65\textwidth]{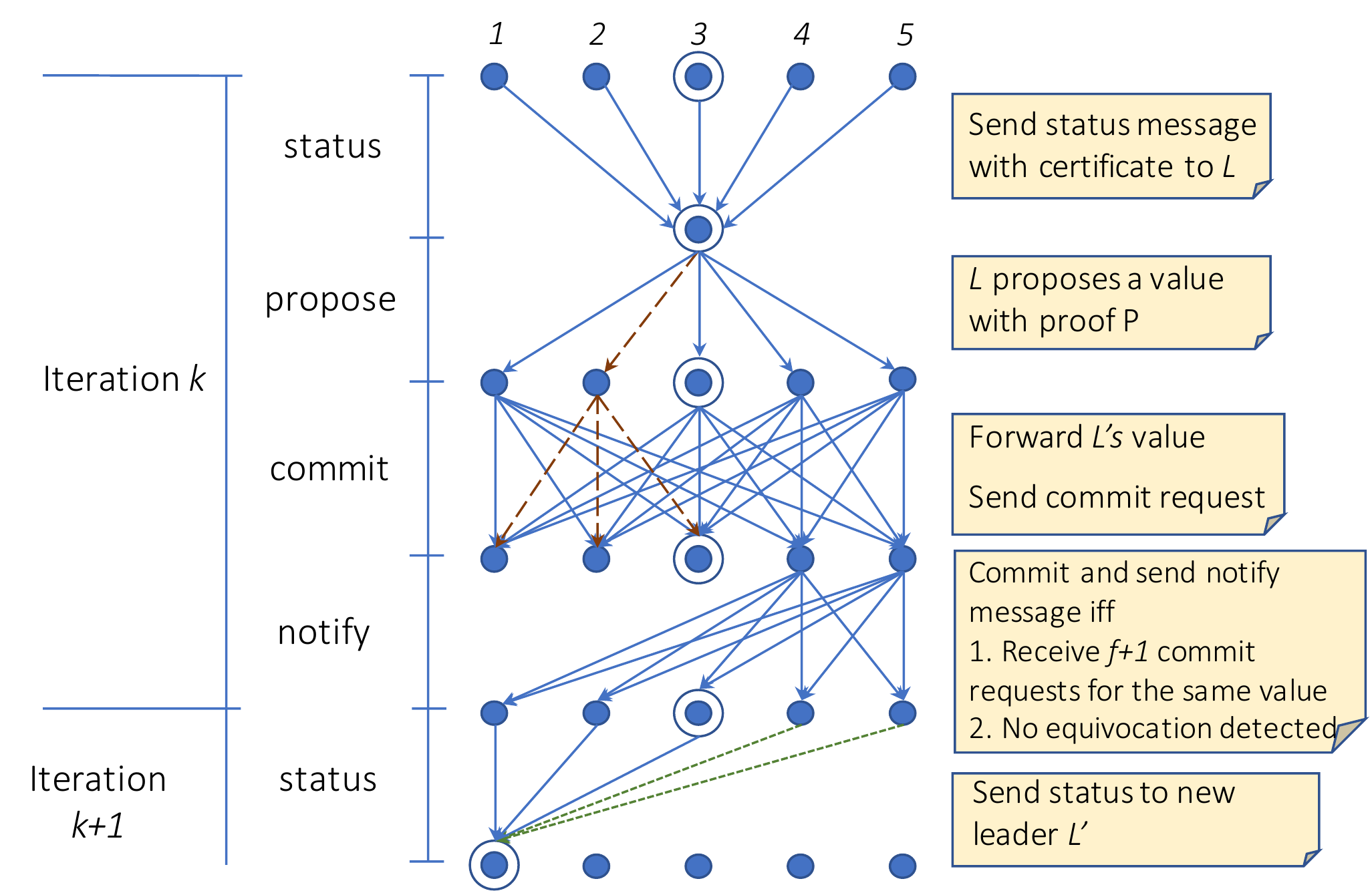}
\caption{\textbf{Synchronous Byzantine Synod Protocol with $n = 2f+1$}. 
The figure shows an example iteration.
In this example, replicas 2 and 3 are Byzantine and replica 3 is the leader $L$ in iteration $k$.
\textbf{0.} $(\Status)$ Each replica sends its current states to $L$. 
\textbf{1.} $(\Propose)$ 
No replica has committed or accepted any value, so $L$ can propose any value of its choice. 
$L$ equivocates and sends one proposal to replica 2 
(shown by dashed red arrow) and a different proposal to honest replicas. 
\textbf{2.} $(\Commit)$ Honest replicas forward $L$'s proposal and send $\Commit$ requests to all replicas.
Replica $2$ only sends to replicas $\{1,2,3\}$.
Replicas 4 and 5 receive $f+1$ commit requests for the blue value and do not detect equivocation, so they commit. 
Replica 1 detects leader equivocation and does not commit despite also receiving $f+1$ $\Commit$ requests for the blue value. 
\textbf{3}. $(\Notify)$  
Replicas 4 and 5 notify all other replicas and terminate. 
On receiving a valid notification, replica 1 accepts the blue value. 
\textbf{4}. $(\Status)$ 
The replicas send $\Status$ messages to the new leader $L'$ for iteration $k+1$. 
$\Status$ messages from terminated replicas 4 and 5 are virtual (shown by dotted green arrows).}
\label{fig:protocol}
\end{figure}

\bigskip

\section{An Improved Protocol with Stable Leaders}	\label{sec:smr:stable_leader}
The improved protocol distinguishes the common case and the view change procedure.
The common case runs the $\Propose$ round, the $\Commit$ round, a ``lightweight'' $\Notify$ round in iterations.
The $\Status$ round and the full $\Notify$ round will be pushed to the view change procedure.
Concurrent to the common case and view change, each replica additionally runs a checkpoint procedure and a leader monitoring procedure. 

\subsection{The Common Case}
Suppose an honest replica $i$ is currently in view $l$ and is working on slot $s$, and the current iteration is $k$.

\begin{itemize}
\addtolength{\itemindent}{32pt}
\item[\textbf{Round 1}] ($\Propose$)
Same as before, except that for slots that have not been worked on in prior views, the leader does not need to include safe value proofs (see view change).

\item[\textbf{Round 2}] ($\Commit$)
Same as before. 

\item[\textbf{Round 3}] ($\NotifyLight$)
If replica $i$ has committed on $v$ at the end of Round 2,
it sends a notification summary $\sig{s, k, \Notify, v}_i$ to every replica including itself. 

At the end of this round, 
if replica $i$ receives $f+1$ valid $\sig{s, k, \Notify, v}_j$ summaries for the same $v$,
it concatenates them to form a $\Notify$ certificate $\NotSum$; 
otherwise, it marks the current leader $L_l$ as faulty.
\end{itemize}

The $\NotifyLight$ round does not make other replicas accept a committed value.
Its purpose is to grant a replica the ability to prove the correctness of the committed value to other replicas or to an external client, using $\NotSum$.

\paragraph{Reply to clients.}
As in Section~\ref{sec:smr:basic}, a client can be convinced about a committed decision if it receives $f+1$ $\Notify$ summaries at the end of Round 3, or wait for one extra round to avoid talking to $f+1$ replicas.

\paragraph{Checkpoint.}
Before presenting the view change procedure, we need to introduce checkpoints.
We say slot $s$ is a stable checkpoint for replica $i$, if replica $i$ knows for sure that every honest replica has committed for slot $s$.
In fact, an easy way to create a stable checkpoint is to broadcast a valid $\Notify$ certificate $\NotSum$ (obtained in Round 3) to every replica. 
If another replica receives a valid $\NotSum$ for slot $s$, it can safety commit and moves to the next slot.
Thus, if replica $i$ broadcasts $\NotSum$ for slot $s$, then it knows for sure every other honest replica will commit for slot $s$ at the end of that round, and slot $s$ now becomes a stable checkpoint for replica $i$.
However, since honest replicas cannot coordinate their actions, this will require an all-to-all round of communication with $\Theta(n)$-sized messages per slot.
Following PBFT, we can create a checkpoint for every batch of slots to amortize the cost of checkpointing.
This involves every replica broadcasting a $\NotifyLight$ message and then a $\Notify$ certificate $\NotSum$ for a digest of a slot batch, e.g., of size $c=100$.
If a scheduled checkpoint does not become stable in time, a replica $i$ accuses the leader.

\subsection{View Change}
The view change procedure is initiated by the new leader $L_{l+1}$ once it obtains a $\Blame$ certificate $\BlameSum$, i.e., $f+1$ $\sig{\Blame, l+1}_i$ messages from distinct replicas.
We henceforth write $L_{l+1}$ as $L'$ for short.

\begin{itemize}
\addtolength{\itemindent}{50pt}
\item[\textbf{Round VC1}]
$L'$ sends a $\sig{\NewView, l+1, \BlameSum, s', \CPSum}_{L'}$ to every replica including itself.
$s'$ is the last stable checkpoint known to $L'$, and $\CPSum$ is a proof for that checkpoint. 

At the end of this round, if replica $i$ receives a valid $\NewView$ message,
it exits view $l$ and sets $\Enter_i := \mathsf{true}$.

\item[\textbf{Round VC2}]
If $\Enter_i = \mathsf{true}$, replica $i$ forwards the $\NewView$ message it receives in Round VC1 to every other replica.

At the end of this round, if replica $i$ is forwarded a valid $\NewView$ message but did not directly receive one from $L'$ in the previous round, or if replica $i$ detects equivocation about $s'$ by $L'$,
then replica $i$ exits view $l$, sets $\Enter_i := \mathsf{false}$, and marks $L'$ as faulty.

\item[\textbf{Round VC3}] ($\Notify$)
For every committed slot since the last stable checkpoint, replica $i$ sends the full $\Notify$ message, which includes the corresponding $\Cert$, to every replica including itself.  

At the end of this round, for every slot uncommitted slot, if replica $i$ receives a valid $\Notify$ message for that slot, it accepts the value in the $\Notify$ message.  

\item[\textbf{Round VC4}] ($\Status$)
Let $T$ be the largest slot that replica $i$ has committed or has accepted a value.
For every slot since the last stable checkpoint up to $T$, replica $i$ sends a $\Status$ message, which includes the corresponding $\Cert$, to the new leader $L'$.
Replica $i$ also sends $L'$ a $\sig{T, l+1, \StatusMax}_i$ message, indicating that it has not committed or accepted any value for any slot numbered greater than $T$.
$L'$ will use these $\StatusMax$ messages to prove that certain slots have not been worked on in prior views. 

At the end of this round, if $\Enter_i=\mathsf{true}$, then replica $i$ increments its view number $l$ and enters the new view. 
From the next round, replica $i$ returns to the common case and works on slot $s'+1$. 
If $\Enter_i=\mathsf{false}$, replica $i$ increments $l$ but does not enter the new view. 
\end{itemize}

It is crucial to note that we distinguish ``having a view number $l$'' and ``being in view $l$''.
A replica only enters a new view (and stay in that view) upon receiving a $\NewView$ message directly from the new leader.
A forwarded $\NewView$ message will only make a replica exit its current view, but not enter the new view.
This ensures the following.

\begin{lemma}
Honest replicas will not be in different views.
\label{lemma:same_view}
\end{lemma}
\begin{proof}
If an honest replica $h_1$ is in view $l$, according to the protocol, it must have received a valid $\NewView$ message from $L_l$.
Then, it must have forwarded the $\NewView$ message to all other replicas.
Thus, all other honest replicas must have exited view $l-1$ in Round VC2 of that view change.
\end{proof}

Lemma~\ref{lemma:same_view} states that honest replicas will never be in different views.
But some honest replicas may not be in any view. 
If a replica is not in any view, it does not participate in Rounds 1 to 3 of the common case.
In particular, it ignores any proposal and does not send any $\Commit$ request or $\Notify$ summary. 
But it will participate in Round 4 of the common case: it will commit a value upon receiving a valid $\Notify$ certificate.

\paragraph{Leader monitoring.}
In the view change procedure, if replica $i$ is forwarded a $\NewView$ message but does not directly receive one from $L'$, it knows $L'$ is faulty. 
Another situation where replica $i$ detects $L'$ as faulty is when replica $i$ has sent $\BlameSum$ to $L'$ but $L'$ does not initiate the view change procedure. 
In both cases, replica $i$ will skip $L'$ and wait for the next view change.
In more details, replica $i$ keeps monitoring its current leader $L_l$ and future leaders as follows.

\begin{itemize}
\item[--] 
At the beginning of any round, if replica $i$ has marked $L_l$ faulty, it sends a $\sig{\Blame, l+1}_i$ message to every replica including itself. 

\item[--] 
At the end of any round, if replica $i$ has gathered $f+1$ $\sig{\Blame, l+1}_j$ messages, it concatenates them to form a $\Blame$ certificate $\BlameSum$ (else $\BlameSum := \bot$).

\item[--]
At the beginning of any round, if a replica $i$ has $\BlameSum \neq \bot$, it sends $\BlameSum$ to the next leader $L'$.

\item[--]
At the end of any round, if replica $i$ has sent $\BlameSum$ to $L'$ in the previous round but does not receive a $\NewView$ message from $L'$ in this round, 
then replica $i$ marks $L'$ as faulty and increments the view number $l$.
\end{itemize}

We remark that we have only pointed out the necessary conditions for an honest replica to accuse a faulty leader.
Honest replicas may detect a faulty leader through other means (e.g., equivocation and invalid messages).
But whether or not they accuse the leader in those situations will not affect the safety and liveness of the protocol.

\subsection{Safety and Liveness}
The proof for safety remains largely unchanged from Section~\ref{sec:proof} except for a small modification to the base case of Lemma~\ref{lemma:cert}.

\newcounter{tmp}
\setcounter{tmp}{\value{theorem}}
\setcounterref{theorem}{lemma:cert}
\addtocounter{theorem}{-1}
\begin{lemma}[restated with slots]
Suppose replica $h$ is the first honest replica to commit for slot $s$.
If replica $h$ commits on $v^*$ in iteration $k^*$ and $\Cert$ certifies $(s,v,k)$ where $k \geq k^*$, then $v=v^*$.
\end{lemma}
\begin{proof}
For the base case, assume for contradiction that a $\Commit$ certificate for $(s,v,k^*)$ exists for $v \neq v^*$. 
An honest replica $h_1$ must have sent a $\Commit$ request for $v$. 
Replica $h_1$ must then be in some view; otherwise, it would not have participated in the $\Commit$ round at all.
Due to Lemma~\ref{lemma:same_view}, replica $h_1$ must be in the same view as replica $h$.
Thus, $h_1$ must have detected leader equivocation and would not have committed.
The proof for the inductive step remains unchanged.
\end{proof}

\setcounterref{theorem}{thm:safety}
\addtocounter{theorem}{-1}
\begin{theorem}[Safety, restated with slots]
If two honest replicas commit on $v$ and $v'$ respectively for slot $s$, then $v=v'$.
\end{theorem}
The proof invokes Lemma~\ref{lemma:cert} in the same way (but with slots), and we do not repeat it.
\setcounter{theorem}{\thetmp}

\bigskip

\begin{theorem}[Liveness]
The protocol commits $s$ slots in at most $3s + O(cf)$ rounds, where $c$ is the checkpoint batch size.
\end{theorem}

\begin{proof}
The proof observes two crucial properties of the protocol: 
(1) an honest leader will not lose its leader role, and 
(2) each Byzantine leader can prevent progress for at most $O(c)$ rounds. 
The theorem then follows from the above properties.
Once an honest leader takes control, it will not be replaced and a new slot is committed in its view every iteration (3 rounds) in its view in the common case. 

For part (1), it suffices to show that an honest replica will not accuse an honest leader. 
In the protocol, there are two situations in which a replica accuses a leader: 
not stepping up for its view and not making progress in its view. 
Neither will happen to an honest leader.
Before accusing a future leader of not stepping up, an honest replica will send a $\Blame$ certificate to the leader.
At this point, an honest leader will broadcast $\NewView$ and make all honest replicas enter its view.
Once all honest replicas enter its view, they will all make progress, receive $f+1$ $\Notify$ summaries, and move to the next slot after every iteration.

For part (2), we first show that the last stable checkpoints of any two honest replicas can be at most off by 1 (i.e., $c$ slots apart).
Let $s_i$ be the last stable checkpoint known to replica $i$.
Suppose for contradiction that two honest replicas $i$ and $j$ have $s_i > s_j + 2c$.
This means at checkpoint $s_j + c$, no honest replica has broadcast $\NotSum$.
Then, all honest replicas should have accused the leader and replica $i$ should not have advanced to slot $s_i > s_j + 2c$, a contradiction.
Therefore, a Byzantine leader can disrupt progress by at most $c$ iterations when its view begins (by starting from a checkpoint $c$ slots behind),
and waste another $c$ iterations (by leading the common case normally until the next checkpoint) before getting accused and replaced.
\end{proof}

\paragraph{Remark:} It should now be clear that the checkpoint batch size $c$ is a trade-off between common case efficiency and worst-case efficiency.
A larger $c$ allows Byzantine leaders/replicas to prevent progress for longer with some carefully planned malicious actions; 
but in the best case where there is no such malicious behavior, a larger $c$ means less frequent checkpoints and hence less work in the common case.


\bigskip
\section{An Improved XFT Variant}	\label{sec:xft2}

XFT relies on an all-honest active group of $f+1$ to make progress~\cite{XFT}.
As we mentioned, when the number of faults $f$ approaches the $(n-1)/2$ limit, i.e., $n = 2f+1$, there is only one all-honest group among the $n \choose f+1$ total groups.
Finding this single group will be a challenge.

The XFT paper~\cite{XFT} mentions the simple scheme of trying out all groups, 
and also acknowledges that this simple scheme does not scale with $n$ and $f$ as it would require a super-exponential number of trials to find the only all-honest group.
Below we describe an XFT variant that only requires a quadratic number of trials to find the all-honest group.
The following description assumes the reader is familiar with the XFT protocol.

We will rotate the leader of the active group in a round robin fashion, and
let the leader pick its own $f$ followers.
Furthermore, we give each leader $f+1$ chances of view changes, i.e., let it re-pick its followers $f$ times, before replacing the leader.
We refer to all views under a leader as that leader's \emph{reign}.

If a leader is honest, in each view during its reign, it either makes progress (if all followers respond to its proposal correctly) or it detects at least one faulty follower that does not respond correctly.
In the latter case, the leader can locally mark that follower as faulty and replaces it with a new replica that has not been marked faulty. 
During an honest leader's reign, there can be at most $f$ view without progress. 
After that, the leader should have locally detected all $f$ faulty replicas and found the all-honest active group.
Thus, if there have been $f+1$ view without progress during a leader's reign, the leader must be faulty and should be replaced. 
Each faulty leader can cause at most $f+1$ view changes, and the $f$ Byzantine replicas combined can cause at most $(f+1)f = O(n^2)$ view changes.


\end{document}